\documentclass{article}
\usepackage[all]{xy}
\usepackage{amssymb}
\usepackage{stmaryrd}
\usepackage{amsmath, amsthm} 
\usepackage{graphicx}
\usepackage{listings}
\usepackage{enumerate}
\usepackage{url}
\urldef\myurl\url{foo%.com}
\usepackage{hyperref}

\newtheorem{definition}{Definition}[section]
\newtheorem{theorem}{Theorem}[section]
\newtheorem{proposition}{Proposition}[section]
\newtheorem{lemma}{Lemma}[section]
\newtheorem{informal-proposition}{Informal proposition}[section]
\newtheorem{problem}{Problem}[section]

\newtheorem{unknown-proposition}{unknown-Proposition}[section]

\newtheorem{example}{Example}[section]
\newtheorem{question}{Question}[section]

\newcommand{\real}{\mathbb{R}}

\newcommand{\sg}{\mathsf{s}}

 \newcommand{\leftmult}{~
     \mathbin{\setlength{\unitlength}{1ex}
     \begin{picture}(1.4,1.8)(-.3,0)
     \put(-.6,0){$\times$}
     \put(-.46,-0.0){\circle{0.4}}
     \end{picture}
     }}

\begin{document}

\title{\bf Expressing entropy and cross-entropy in expansions of common meadows\\}

\author{Jan A Bergstra \\
Informatics Institute, University of Amsterdam,\\
 Lab 42, Science Park 900, 1098 XH, Amsterdam, 
The Netherlands\\
j.a.bergstra@uva.nl\\ \and \\
John V Tucker\\
Department of Computer Science,\\ Swansea University, Bay Campus, \\Fabian Way, Swansea, SA1 8EN, United Kingdom\\j.v.tucker@swansea.ac.uk }

\maketitle

\abstract{\noindent A common meadows is an enrichment of a field with a partial division operation that is made total by assuming that division by zero takes the a default value, a special element $\bot$ adjoined to the field.   To a common meadow of real numbers we add a binary logarithm $\log_2(-)$, which we also assume to be total with $\log_2(p) = \bot$ for $p \leq 0$.  With these and other auxiliary operations, such as a sign function $\sg(-)$, we form algebras over which entropy and cross entropy can be defined for probability mass functions on a finite sample space by algebraic formulae that are simple terms built from the operations of the algebras and without case distinctions or conventions to avoid partiality.  The discuss the advantages of algebras based on common meadows, whose theory is established, and alternate methods to define entropy and other information measures completely for all arguments using single terms.}
\medskip


\section{Introduction}

\subsection{Information theory.}
We will investigate how entropy and cross-entropy, and other information theoretic qualities, can be defined uniformly and completely for all arguments by formulae or algebraic expressions over algebras of reals numbers.  By a formula, or algebraic expression, we have in mind a single term built from names for the basic operators of an algebra. 

Consider the entropy of a probability mass function $P$ on a finite sample space $V$ as given by 
$$H(P) = \sum_{x \in V}(P(x) \cdot \log_2 \frac{1}{P(x)}).$$
Now, in case $P(x)=0$ for some $x \in V$, the expression for entropy contains a subterm $\frac{1}{P(0)}$ that is not defined.  Since $P(x)=0$ is a reasonable value the definition is not complete. An alternate formulation for entropy is
$$H(P) = -\sum_{x \in S}(P(x) \cdot \log_2 P(x)).$$
Here,  $P(x)=0$ means that  the subterm $\log_2 P(x)$ is not defined. 

In most practice, the facts that the formulae do not completely define the measures are not considered problematic.  Rather than introduce case distinctions, these definitions are finished off -- indeed, the second definition is preserved -- by adopting a convention, namely that $0 \cdot \log_2 0= 0$. Note that such a convention is an equation that might be true or false of some underlying algebra of reals with the operators $\cdot$ and $\log_2$.

A closely related concept is cross entropy defined for two probability mass functions, say $P$ and $Q$, by
$$H(P,Q) = \sum_{x \in S}(P(x) \cdot \log_2 \frac{1}{ Q(x)})$$
Alternately, it can be formulated as
$$H(P,Q) = -\sum_{x \in S}(P(x) \cdot \log_2  Q(x)).$$
In most practice, conventions in use here include the use of a special element added to the real numbers, namely $\infty$.

The enforcement of conventions such as an equation $0 \cdot \log_2 0= 0$ and the addition of a new element $\infty$, have implications for the algebra and logic of the information measures. Here we seek answers to the question:

\begin{question}\label{term_definability_problem}
What algebras of real numbers would allow entropy and cross entropy to be defined uniformly and completely for all arguments by formulae that are single terms over the operations of the algebras? 
\end{question}

Note that since $H(P) = H(P,P)$, we expect if cross entropy is definable then so is entropy.

We have explored this question of definability elsewhere. In \cite{BergstraT2025Entropy}, we developed an algebra $A$ of real numbers in which the convention $0 \cdot \log_2 0= 0$ is an equation true of $A$; this algebra $A$ we called the \textit{entropic transreals}. The entropic transreals are a modified version of the establised transreals of~\cite{Anderson2007VG,AndersonVA2007VG,DosReisGA2016IJAM,DosReis2019TM,BergstraT2020TM}. Unfortunately, entropic transreals have non-associative multiplication, and lack an important algebraic property of division (\textit{viz.} fracterm flattening), both facts giving reasons for looking for expressions of entropy and cross-entropy in other structures for real arithmetic with more conventional properties.  


Here, we will examine the problem in a fundamental way,  by choosing to start with common meadows, whose theory is established and
based on insights into the semantics, algebra and logic, and practical use of abstract data types for computer calculation.


\subsection{Totality and term definability}\label{Totality_and_definability}

Specifically, the technical problems arise because the key operations are partial:  $\frac{x}{0}$ is not defined and $\log_2{x}$ is not defined for $x\leq 0$. Thus, to avoid these undefined arguments, which disturbs the formulae and suggests problem \ref{term_definability_problem}, a conditional branching operator seems necessary. 

Our previous studies of totality \textit{v.} partiality in algebras have focussed on division $\frac{x}{y}$: by adding this partial operator to the field of reals we obtain the \textit{meadow} of real numbers \cite{BergstraT2007JACM}. Having studied a number of semantic options for making partial functions total (see \cite{BergstraT2023}), one has stood out. 
We add a special element $\bot$, define $\frac{x}{0} = \bot$, and how $\bot$ interacts with the operations on the reals. The standard approach to the algebra of $\bot$ is to make it absorptive. This constructs a \textit{common meadow} of real numbers. 
Thus, common meadows are our starting point for this problem. 

Division is so basic it can be considered \textit{primus inter pares} among functions that need to be added to the field operations. For the purposes of our problem with information measures, we must add $\log_2$ to the common meadow of real numbers. This produces an algebra of the form
$$\real_{ \bot,\log_2} = (\real \cup \{ \bot \} \ | \ 0, 1, \bot, x+y, -x, x \cdot y, \frac{x}{y}, \log_2(x))$$
Here, to make the algebra total, we set $\log_2(x) = \bot$ for $x\leq 0$.  Clearly, this algebra has the all operators that appear in the standard formulae above. We confirm there is a problem (Proposition \ref{basic_incompleteness} below):
\medskip

\noindent \textbf{Theorem.} {\em The algebra $A = \real_{ \bot,\log_2}$  {\em fails} to allow entropy $H(P)$ for a sample space of size $n=2$ to be defined by a term over $A$}. 
\medskip
 
Then, to search for uniform and complete formulae for entropy and cross entropy, we will add more operations $f_1, \ldots, f_n$ to this basic algebra $\real_{ \bot,\log_2}$ and examine the problem for several algebras of the form
$$\real_{ \bot,\log_2} = (\real \cup \{ \bot \} \ | \ 0, 1, \bot, x+y, -x, x \cdot y, \frac{x}{y}, \log_2(x), f_1, \ldots, f_n).$$
Among the additional operations, we focus on adding a conditional operator, a new multiplication operator, and a sign function, which simulates an ordering.  


\subsection{Structure of the paper.}
In Section \ref{common_meadows}, we discuss the method of totalising partial operations with $\bot$, whose semantics is open to informal interpretations, and we describe the first algebras of reals we use. 

In Section \ref{terms_etc}, we discuss some algebraic points, and using a conditional operator, and a multiplication operator derived from it, will give term expressions for entropy and for cross-entropy. 

In Section \ref{needed_auxiliary_operators}, we show that operators need to be added to $\real_{ \bot,\log_2}$.

In Section \ref{general_approach}, we re-define the  approach and examine some auxiliary operators, notably a sign function $\sg$.

In Section \ref{observations} we look at other information measures.

In Section \ref{flattening} we consider the algebraic properties of flattening.

In Section \ref{concluding_remarks} we consider some matters arising from the paper.


\section{Totality and the common meadows}\label{common_meadows}

We summarise the algebraic ideas we need.


\subsection{On the default peripheral value $\bot$}\label{default_peripheral_value}

Consider an algebra enlarged with a special element $\bot$, which may be read as an `unconventional outcome'. We will refer to it as a \textit{default peripheral value} (DPV); `default'  because it can be used as an output by operators that fail to give a meaningful value for some reason, and as `peripheral' as it is extraneous to the algebra.  Mathematically, its primary role is to totalise functions that otherwise would be partial.  However, on introducing $\bot$ to an existing algebra we must define how  $\bot$ interacts with its operations.  A general method  is to make $\bot$ \textit{absorptive}, adapting the operations to return $\bot$ whenever $\bot$  is an argument. This general technique is worked out in \cite{BergstraT2021b}.

The advantage of total functions over partial functions is the relative simplicity of working with the algebra and logic of total functions, in comparison with working with algebras and logics for partial functions. Appreciation of this advantage can be a controversial matter, but we hold that a first order logic of total functions is significantly simpler and applicable than any of its various counterparts involving logics of partial functions.

The general idea of a DPV is intended to be adaptable to contexts. Now, many functions that are partial acquire conventions to complete their definition that are not seen to be \textit{ad hoc}. Certainly, the element $\infty$ is commonly employed: for example, since we know $\lim \frac{1}{x} \to \infty$ as $x \to 0$ we have a reason to choose to totalise division by setting $\frac{1}{0} = \infty$. Further, we can have the intuition that an infinity can be absorptive: operations return $\infty$ whenever $\infty$ is an argument.  However, adding a new element to an algebra with some desired intuition requires careful consideration of how the new element interacts with all the operations of the algebra. So, if $\infty$ is a renaming of $\bot$ then we define $-\infty = \infty$ (and we do not have $+\infty$ and $-\infty$). Using this $\infty$ instead of $\bot$ is just \text{one} form of infinity; it can be called a \textit{pure absorptive infinity}. Infinity can interact with the operations in other ways even in a common meadow: we will return to this topic in Section \ref{peripherals4infinity}.

Clearly, in such circumstances we might choose to use the symbol $\infty$ to do the algebraic work that $\bot$ is designed to do, \textit{viz.} when a natural ordinary value for a term $t$ cannot be found we have $t = \bot$.  To consider absorptive infinity in place of $\bot$ is an option for the definitions of entropy and cross entropy in this paper, since the symbol $\infty$ appears in some conventions to be found in information theory. It is not an option we will take.

Below we will use $\bot$ instead of $\infty$ as all arguments are best understood with the intuition that $f(a) = \bot$ expresses that within $\real$, $f(a)$ is undefined. This notational difference is strictly speaking unnecessary because as a function $f(-)$ is immune for a change of name of $\infty$. We will use $H_\bot(P)$ and $H_\bot(P,Q)$ for entropy and cross-entropy but such that the outcome will be $\bot$ rather than $\infty$ if the outcome is a peripheral number.

In~\ref{peripherals4infinity}, we provide a brief survey of peripherals for infinite values, with a naming convention for these peripherals. 



\subsection{An algebra of real numbers for entropy definitions}\label{algebra_real_numbers}

There are dozens of functions on the real numbers that are essential for applications and yet are naturally partial. 
A most obvious case is division $\frac{x}{y}$: by adding  this as a basic operation to the field we obtain the \textit{meadow} of real numbers \cite{BergstraT2007JACM}. Since division is partial at $x=0$ we add the peripheral value $\bot$ to make it total, i.e., $\frac{x}{0} = \bot$. 
Assume that a field $\real$ of real numbers is given with the classic operations \cite{vanderWaerden1970}:
$$x+y, \ x\cdot y, \ -x.$$ 
The peripheral $\bot$ interacts with these operations and division using the general absorptive method, adapted to the operations, such that 
$$-\bot = x+\bot = \bot + x = x \cdot \bot = \bot \cdot x =  \frac{x}{\bot} = \frac{\bot}{x} = \bot$$ 
This constructs the \textit{common meadow} of real numbers:
$$\real_{ \bot} = (\real \cup \{ \bot \} \ | \ 0, 1, \bot, x+y, -x, x \cdot y, \frac{x}{y}).$$
Common meadows as a general construction for fields were introduced in~\cite{BergstraP2015LNCS}. The equational axiomatisations and the equational theory of common meadows is analysed in~\cite{BergstraT2022TCJ,BergstraT2023,BergstraT2024a}.

To proceed, we now add the partial function $\log_2(x)$.
$$\real_{ \bot,\log_2} = (\real \cup \{ \bot \} \ | \ 0, 1, \bot, x+y, -x, x \cdot y, \frac{x}{y}, \log_2(x)).$$

%
%

The problem we posed in the Introduction, Question \ref{term_definability_problem},  can be stated precisely as follows:
\begin{question}\label{question1}
Can the information-theoretic measures of entropy and cross entropy be defined uniformly for all inputs as single terms over this algebra $A = \real_{ \bot,\log_2}$ of reals, or over some of its expansions?
\end{question}


\section{Defining information measures as terms}\label{terms_etc}

Here we make a first attempt at defining the measures.


\subsection{Preliminaries on terms}\label{preliminaries}

Our work in this paper is founded upon concepts and mehods from abstract data type theory: signatures, terms, equations etc. \cite{EhrichWL1997}. The general question in the Introduction is a natural question to do with data types, for the technical Question \ref{question1} has the classic form of a hidden operator problem in data type theory.

However, in order to keep close to the familiar notions and formulae of entropy, we will work without being explicit about syntax and semantics, which are core to analysing data types. The users of information theory are comfortable with this informal way.  At some places it may be helpful to remember the role of syntax in making arguments about terms more precise; we trust the step-up in formality will be easy when needed. 

Now, mindful of the classic formulae, we must make a remark on general summations. Let $t$ be a term over a signature with $+$; suppose $+$ is both commutative and associative. The construction $\sum_{x  \in V} t$ is available for any finite set $V$ as follows: 

\begin{definition}
For a given finite set $V$ of cardinality $n = |V|$, with enumeration
without repetitions  $a_1,\ldots, a_n$, with tha $a_i$ closed terms, the semantics of $\sum_{x  \in V} t$ is given by the term
$$\sum_{x  \in V} t = [a_1/x]t + [a_2/x]t +\ldots +[a_n/x]t.$$
\end{definition}

This construction is independent of the choice of the enumeration because addition is both commutative and associative. We will refer to this construction as a {\em finitary generalized sum}. 


\subsection{Additional operators: conditional and left-sequential multiplication}\label{two_conditional_operators}
One approach to integrating the exceptional case of $P(x) = 0$ is to add operations to the algebra of reals that can create a larger class of terms with which case distinctions can be made. 

\begin{definition}
A {\em conditional operator}  $ x \lhd y \rhd z$ uses as the test $y=0$ and is defined by:

if $y = 0$ then $ x \lhd y \rhd z=z$; 

if $y \neq 0$ and $y \neq \bot$ then 
$x \lhd y \rhd z=x$; and

if $y = \bot$ then $ x \lhd  \bot \rhd z= \bot$.
\end{definition}

Notice that the conditional operator compromises the absorptive character of $\bot$ as $\bot\! \lhd 0 \rhd 1 = 1$ rather than $\bot\! \lhd 0 \rhd 1 = \bot$.

Using the conditional we can define a new `multiplication' operation on the reals. 
\begin{definition}
The operator $x \leftmult y$, which we call the {\em left-sequential multiplication}, or simply {\em sequential multiplication},  is defined by
$$ x \leftmult y = (x \cdot y) \lhd x \rhd 0.$$
\end{definition}

The idea is that the left (first) argument of a product is evaluated first and is used to obtain a result.
If $x = 0$ then $ x \leftmult y = 0$, or if $x \neq 0$ then $ x \leftmult y = x \cdot y$.


\subsection{Entropy and cross-entropy terms with conditionals}\label{conditionals}

\begin{definition}
A probability mass function on a sample space $V$ is a function $P: V \to \real$ with (i) $\forall_{x \in V}P(x) \geq 0$ and (ii) $\sum_{x  \in V} P(x) = 1$. 
\end{definition}

With the algebraic operations of \ref{two_conditional_operators}, we construct the algebra
$$\real_{ \bot,\log_2,\leftmult} = (\real \cup \{ \bot \} \ | \ 0, 1, \bot, x+y, -x, x \cdot y, \frac{x}{y}, \log_2(x), x \leftmult y).$$
It is easy to check that with a conditional to hand the classical formulae can be made into terms, notably using the operator $\leftmult$:

\begin{proposition}
The Shannon entropy  $H(P)$ of a probability mass function $P$ can be defined as terms over $A$
 $$H(P)= \sum_{x\in V}^{}P(x) \leftmult (\log_2 \frac{1}{ P(x)})$$
Equivalently, Shannon entropy can be defined by the term:
 $$H(P)= -\sum_{x\in V}^{}P(x) \leftmult \log_2  P(x)$$
 Cross entropy is given by:
  $$H(P,Q)= \sum_{x\in V}^{}P(x) \leftmult \log_2 \frac{1}{ Q(x)}$$
  or equivalently: 
 $$H(P,Q)= -\sum_{x\in V}^{}P(x) \leftmult \log_2  Q(x).$$
\end{proposition}

\begin{proof}
We consider one of the measures for illustration, in particular the second definition of entropy:
$$H(P)= 
-(\sum_{x\in V, P(x) >0}^{}P(x) \cdot \log_2  P(x)~
+\sum_{x\in V, P(x) =0}^{}0 ) 
=-\sum_{x\in V}^{}P(x) \leftmult \log_2  P(x).$$
\end{proof}

This fact is an algebraic tool in what follows. Returning to the discussion of the neutral role of $\bot$ in Section~\ref{default_peripheral_value} and its possible renaming by $\infty$, let us look at some examples.
 
\begin{example}
{\em With $\real_{\infty_\mathsf{cmu}/\bot}$ we denote a common meadow of reals where $\bot$ has been renamed into $\infty_\mathsf{cmu}$. 
Upon expansion with $\log_2$, and with left sequential multiplication $\leftmult$, we obtain the structure $\real_{\infty_\mathsf{cmu}/ \bot,\log_2,\ \leftmult}$. Assume that for 
$$p \leq 0, \  \log_2 p = \infty_\mathsf{cmu}.$$ 
Then $\leftmult$ works as follows: 
\begin{center}
$a \leftmult b = 0$ if $a = 0$ and $a \leftmult b = a \cdot b$ if $a \neq 0$. 
\end{center}
We note that under these conditions,  $0 \leftmult \infty_\mathsf{cmu} = 0$ so that $\infty$ is not fully absorptive. }
\end{example}

\begin{example}\label{example2}
 {\em Suppose $P$ and $Q$ are probability mass functions with the property that for some $x \in V$, $P(x) = \frac{1}{2}$ and $Q(x)= 0$. Then 
 $$P(x) \leftmult \log_2 \frac{1}{ Q(x)} = \frac{1}{2} \leftmult \log_2 \frac{1}{ 0} = \frac{1}{2} \cdot \log_2 \frac{1}{ 0}= 
   \frac{1}{2} \cdot \log_2 \infty_\mathsf{cmu}= \frac{1}{2} \cdot  \infty_\mathsf{cmu} = \infty_\mathsf{cmu}$$ 
so that $H(P,Q) = \infty_\mathsf{cmu}$. 

It follows that the presence of an infinite value $\infty_\mathsf{cmu}$, 
(or $\bot$ for that matter) in a common meadow of reals, 
is not only relevant for totalising division and logarithm but that totalisation also   
matters for the expression of cross-entropy, because, according to the above calculation, without use of $\infty_\mathsf{cmu}$ (or $\bot$) 
cross-entropy will be a partial function.}
\end{example}

For a function $F : V \to \real$ from a finite collection of atomic events $V$  and a probability mass function $P$, the sequential expected value $\overline{E}^V_{P}(F)$ is given by:
  $$\overline{E}^V_{P}(F)= \sum_{x\in V}^{}P(x) \leftmult F(x).$$

 \noindent We notice that sequential expected value thus defined has the following properties:

(i) if for some $ x\in V$, $F(x) = \bot$ while $P(x) \neq 0$ then $\overline{E}^V_{P}(F)= \bot$;

(ii) otherwise $\overline{E}^V_{P}(F)$ has an ordinary value (i.e., a value different from $\bot$).
 
 Some readers will know these formulae with $\infty$ in place of $\bot$.
 We notice that the various expressions for entropy and cross-entropy  are each instances of
    sequential expected value; for instance, in the case of entropy:
  $$H(P)= \overline{E}^V_{P}( \log_2 \frac{1}{ P(x)}).$$


\section{The need for auxiliary operators}\label{needed_auxiliary_operators}

A meadow extended with the addition of $\log_2$ has all the functions needed to make the classical information formulae, but the algebra $\real_{ \bot,\log_2}$ cannot avoid the partiality problem, even with its total operations.

\begin{proposition}\label{basic_incompleteness}
Let $A =\real_{ \bot,\log_2}$. Then $A$ does not allow entropy $H^V_\bot(P)$ of a probability mass function $P$ for a sample space $V$ of size $n=2$ to be defined by a term over $A$. 
\end{proposition}

\begin{proof} 
Let the sample space $V= \{0,1\}$. Then 
$$H^V_\bot(P) = P(0) \leftmult \log_2 P(0) + (1-P(0)) \leftmult \log_2 (1-P(0)).$$
Introducing variable $p$, an expression for entropy for size $n=2$ amounts to the existence of 
an term $t$ over the signature $\Sigma$ of $A$ such that 
$$t(p) =  p \leftmult \log_2 p + (1-p) \leftmult \log_2 (1-p),$$
for all $p \in [0,1]$. Suppose such a term exists.

We now consider the situation in terms of conventional analysis, without the element $\bot$ and where division and logarithm are partial functions on real numbers. Then, viewing $t$ as a partial function on $\real$ one notices that it is made up of a composition of functions, all of which are differentiable on all arguments in their respective domains. In particular, both for division and $\log_2$ we notice that the domain is an open subset of $\real$, not including $0$, and within the respective domains both functions are continuous and derivable. 

It follows that $t(-)$ understood as a partial function on $\real$ is differentiable whenever defined. In particular $t(p)$, which takes value $0$ for $p=0$ by definition of entropy, must be differentiable for $p=0$. The function 
$h(-) = (1-p) \leftmult \log_2 (1-p)$ is differentiable for $p=1$ so it follows that $t(-) - h(-)= p \leftmult \log_2 p$ is differentiable in $0$.

Thus, it follows that 
$$\lim_{h\downarrow 0} \frac{(p+h) \leftmult \log_2 (p +h) - p \leftmult \log_2 p}{h}$$
must exist within $\real$ for $p= 0$ so that 
$$\lim_{h\downarrow 0} \frac{h \leftmult \log_2 h - 0 \leftmult \log_2 0}{h} = 
\lim_{h\downarrow 0} \frac{h \cdot \log_2 h - 0 }{h} = \lim_{h\downarrow 0} \frac{h \cdot \log_2 h }{h} =
\lim_{h\downarrow 0} \log_2 h  $$ 
exists within $\real$, which is not the case, thereby completing the proof.
\end{proof}


\section{Expressions for entropy and cross-entropy}\label{general_approach}

In Section \ref{conditionals}, we have seen that over $\real_{\bot,\log_2,\,  \leftmult}$ terms exist that define entropy $H^V_\bot(P)$  of $P$,  and cross entropy $H^V_\bot(P,Q)$ of  
$P$ and $Q$, for probability mass functions $P$ and $Q$ over a finite sample space $V$, given a finite enumeration 
without repetition of the elements of $V$, represented by closed terms $t_1,\dots ,t_n $.

To pursue our term definability problem further with other (more familiar) mathematical operators, we set up a general format.


\subsection{Expansion problem for $\real_{\bot,\log_2}$}

\begin{definition}
Let $A =\real_{\bot,\log_2,f_1,..,f_k}$ be an expansion of $\real_{\bot,\log_2}$ with total functions $f_1,\ldots,f_k$ and having signature $\Sigma$. We say that $A$ allows term expressions for entropy for sample spaces of size $n$ if, given a sequence of new constants $c_1, \ldots, c_n$  there is a closed expression $\phi^n_H(\alpha)$, over the signature 
$\Sigma \cup \{c_1,\ldots,c_n\}$, with $\alpha$ a function variable, such that for every subset $V$ of $\real$, and for every interpretation of the new constants that constitutes an enumeration without repetitions of $V$, and for every probability mass function $P$ on $V$, 
$$H^V_\bot(P) = [P/\alpha]\phi^n_H(\alpha).$$ 
\end{definition}

\begin{definition} 
Let $A =\real_{\bot,\log_2,f_1,..,f_k}$ be an expansion of $\real_{\bot,\log_2}$ with total functions $f_1, \ldots, f_k$ and having signature $\Sigma$.  We say that $A$ allows expressions  for cross-entropy for sample spaces of size $n$ if there is an expression $\phi_H(\alpha,\beta)$ over $\Sigma \cup \{c_1,\ldots, c_n\}$  with $\alpha$ and $\beta$ a function variables, such that for every subset $V$ of $\real$, and for every interpretation of the new constants  that constitute an enumeration without repetitions of $V$, and for every pair of probability mass functions $P$,$Q$ on $V$, 
$$H^V_\bot(P,Q) = [P/\alpha,Q/\beta]\phi^n_H(P,Q).$$
\end{definition}

For sample space size $n=1$ the situation trivialises and not even $\log_2$ is needed to express entropy as well as cross-entropy.
However, these questions are left unanswered:

\begin{problem}
Suppose that $A =\real_{\bot,\log_2,f_1,..,f_k}$ allows expressions for entropy on sample spaces of size $n>1$, must $A$ also allow expressions for entropy on sample spaces of size $n+1$?
\end{problem}

\begin{problem} 
Suppose that $A =\real_{\bot,\log_2,f_1,..,f_k}$ allows expressions for cross-entropy on sample spaces of size $n>1$, must $A$ also allow expressions for cross-entropy on sample spaces of size $n+1$?
\end{problem}

As $H^V_\bot(P) = H^V_\bot(P,P)$, if $A$ allows expression of cross entropy on sample spaces of size $n$, then $A$ will also allow expression of entropy on sample spaces of size $n$.


\subsection{Composite operators}

\begin{proposition} 
Let $A =\real_{ \bot,\log_2,f}$, with $f$ defined by 
$$f(x) =  x \leftmult \log_2 x,$$
then $A$ allows term expressions for entropy for all finite sizes of sample spaces $V$. 
\end{proposition} 

\begin{proof} 
Let $n$ be a size and let $c_1,\ldots,c_n$ a corresponding sequence of new constants, then take 
$\phi^n_H(\alpha)= -\sum_{i=1}^{n}f(\alpha(c_i))$. For a subset $V= \{c_1,\ldots,c_n\}$ of size $n$ of $\real$ and a 
probability mass function $P$ on $V$ we find:
$$\phi^n_H(P)= [P/\alpha](-\sum_{i=1}^{n}f(\alpha(c_i)))= -\sum_{i=1}^{n}f(P(c_i))=$$
$$-\sum_{i=1}^{n}P(c_i) \leftmult \log_2 P(c_i) = H^V_\bot(P)$$
\end{proof}

\begin{proposition} 
Let $A =\real_{ \bot,\log_2,f}$, with $f$ defined by
$$f(x) = x \leftmult x,$$
then $A$ fails to allow a term expression for cross-entropy for sample space $V$ of size $n=2$. 
\end{proposition}

\begin{proof} 
We consider the operations of $A$ as partial functions on $\real$. Then notice that all the operations of $A$ are continuous when defined. Thus, every term over the signature $\Sigma$ of $A$  determines a function which is continuous when defined. Now consider $V= \{0,1\}$ and let $P$ and $Q$ be probability mass functions on $V$ as follows: 
$$P(0) = p, P(1) = 1-p, Q(0) = 0, Q(1) = 1.$$
We find 
$$H^V_\bot(P,Q) = p \leftmult 0 + (1-p) \leftmult 0 = p\leftmult 0.$$ 
Let $t(p)$ be an expression for $H^V_\bot(P,Q)$ as a function of $p$. Then it must be the case that $t(0) = 0$ while $t(p) = \bot$ for $t\in (0,1]$ which is in contradiction with the continuity of $t(p)$ in $0$.
\end{proof}


\subsection{Sign function}\label{adding_sign}

We expand $\real_{ \bot,\log_2}$ with a sign function that simulates an ordering \cite{BergstraT2024TM}.

\begin{definition}
A {\em sign operator}  $\sg$ represents an ordering of the meadow $\real_\bot$ of reals, and is defined by:

if $x = 0$ then $ \sg(x) = 0$; 

if $x > 0$ then $ \sg(x) = 1$ 

if $x < 0$ then $ \sg(x) = -1$

if $x = \bot$ then $\sg(x) = \bot$
\end{definition}

\begin{proposition} 
Let $A =\real_{ \bot,\log_2,f}$, with $f$ defined by
$$f(x) = \sg^2(x).$$ 
Then $A$ allows term expressions to define entropy for sample spaces of all finite sizes.
\end{proposition} 

\begin{proof} 
Let $|V| = n$ and $c_1,\ldots,c_n$ a sequence of new constants that enumerates the elements of $V \subseteq \real$ without repetition. 
Now define terms $t_i$ for $i \in [1,n]$ as follows: 
$$t_1 =_{\mathsf{syn}} c_1, \ \ t_{i+1} =_{\mathsf{syn}} f(\alpha(c_{i+1})) \cdot t_{i+1} + (1-f(\alpha(c_{i+1}))) \cdot t_i .$$ 
The idea of this definition is that for all $\alpha$ taking non-$\bot$ values on all elements of $V$, $t_n$ takes the value of $t_i$ with $i$ maximal so that $\alpha(c_i) \neq 0$. 
Then take 
$$\phi^n_H(\alpha)= -\sum_{i=1}^{n}(\alpha(c_i) \cdot \log_2(f(\alpha(c_i)) \cdot \alpha(c_i) + (1-f(\alpha(c_i))) \cdot t_n).$$ 
Given a probability mass function $P$ on $V$ we find 
$$[P/\alpha](\phi^n_H(\alpha)) = -\sum_{i=1}^{n}(P(c_i) \cdot \log_2(\sg^2(P(c_i)) \cdot P(c_i) + (1-\sg^2(P(c_i))) \cdot t_n).$$ 
To see that this works, notice that if $P(c_i) \neq 0$, $\sg^2(P(c_i)) = 1$ so that 
$$\sg^2(P(c_i)) \cdot P(c_i) + (1-\sg^2(P(c_i))) \cdot t_n)= P(c_i)$$
which is non-zero, so that $\log_2 P(c_i) \neq \bot$; and  if $P(c_i) = 0$, $\sg^2(P(c_i)) = 1$ and 
$$\sg^2(P(c_i)) \cdot P(c_i) + (1-\sg^2(P(c_i))) \cdot t_n)= P(t_n)$$
which is guaranteed to be non-zero by the definition of the terms $t_i$, so that 
$$P(c_i) \cdot \log_2(f(P(c_i)) \cdot P(c_i) + (1-f(P(c_i))) \cdot t_n)= 0$$ 
as intended for the definition of entropy on a sample $s$ with $P(s) = 0.$
\end{proof} 
	
\begin{proposition} 
Let $A =\real_{ \bot,\log_2,f}$, with $f$ defined by 
$$f(x) = \sg^2(x).$$
Then $A$ allows term expressions to define cross-entropy ffor sample spaces of all finite sizes.
\end{proposition} 

\begin{proof} 
The proof follows the same line as above with the following modification:
$$\phi^n_H(\alpha,\beta)= -\sum_{i=1}^{n}(\alpha(c_i) \cdot \log_2(f(\alpha(c_i)) \cdot \beta(c_i) + (1-f(\alpha(c_i))) \cdot t_n)$$
This expression for cross-entropy takes care of peripheral outcomes arising in case $Q(c_i) = 0$ while $P(c_i) \neq 0$. 
\end{proof}

Her are some observations about $\sg^2(x)$:

\begin{lemma}
Whilst $\sg^2(x) = x \leftmult \frac{1}{x}$, the function $f(x, y) = x \leftmult y$ has no term definition over
$A =\real_{\bot,\log_2,\sg^2}$. 
\end{lemma}

\begin{proof}
To see the latter, notice that $0 \leftmult \bot = 0$ so $\leftmult$ is not absorptive. However, for all operations of $A$ including $\sg^2$, $\bot$ is absorptive, i.e., $\bot$  is propagated. 
\end{proof}

\begin{lemma}
The conditional operator $x \lhd y \rhd z$ can be defined in terms of $\leftmult$ by
$$x \lhd y \rhd z= \sg^2(y) \leftmult x+ (1 - \sg^2(y)) \leftmult z.$$
Conversely, $\leftmult$ is defined in terms of the conditional operator:  $x \leftmult y = (x \cdot y) \lhd x \rhd 0$.
\end{lemma}


\section{Expansions expressive for cross-entropy }\label{observations}

By a cross-entropy expressive expansion of 
$A =\real_{ \bot,\log_2}$ we mean an expansion for which term definitions can be given for cross-entropy for all finite sample sizes. Clearly, such algebras exist as adding  $\sg^2$  to  $A =\real_{ \bot,\log_2}$ suffices. There are others.

Consider the following three observations: 

(i) A probability mass function will not take any negative value on any element of the sample space. 

(ii) A probability mass function will not take value $\bot$ on any element of the sample space. 

(iii) For $y \geq 0$,  
$$ \log_2 y = \frac{\log_2 y^2} {2}.$$

\begin{proposition}\label{Prop-on-f}
Let 
$$f(x,y) = x \leftmult (\frac{\log_2 y^2)}{2}) + 0 \cdot y.$$
Then $A =\real_{ \bot,\log_2, f}$ allows term definitions of entropy, and of cross-entropy, for sample space of all finite sizes.
\end{proposition}

\begin{proof}
We notice first that in the definition of $f(x,y)$ the summand $0 \cdot y$ is included in order to have this function $\bot$-preserving in both arguments. The proof of this proposition is immediate as in the defining expressions of the kind 
$$H(P,Q) = -\sum_{x \in S}(P(x) \cdot \log_2  Q(x))$$
for entropy and cross-entropy, in occurrences of
$t \cdot \log_2 r$, $r$ is  non-$\bot$ (observation (ii)) and non-negative 
(observation (i)) so that $t \cdot \log_2 r$ may be replaced by 
$t \cdot \frac{\log_2 r^2}{2}+ 0 \cdot r$. 
\end{proof}
The function $f(x,y)$  can be expressed with the help of $\sg^2$ and $\log_2$ as follows:
$$ f(x,y) = x \cdot\frac{\log_2 (y^2 + 1 - \sg^2(x))}{2}.$$

\begin{problem} 
\label{P1}
With $f$ as above, does there exist an term definition of $\sg^2(x)$
in $A =\real_{ \bot,\log_2,f}$?
\end{problem} 

Whether or not an expansion of $A =\real_{ \bot,\log_2}$ has is expressive for cross-entropy  may or may not boil down to the possibility to express some finite collection of $\bot$-preserving  functions with one or more variables. Perhaps $\{\sg^2\}$ provides such a finite collection? 
If so, Problem~\ref{P1} has a positive solution, of not perhaps $\{f\}$ is such a finite collection, an option which leads to the following question:
\begin{problem} 
Suppose that $A =\real_{\bot,\log_2,f_1,..,f_k}$ allows term definitions for cross-entropy on sample spaces of size $n$, for all $n$. Then, must $A$ also allow a term definition of the function $f(x,y)$ appearing in the statement of Proposition~\ref{Prop-on-f}?
\end{problem}

We expect that providing an explicit definition of $\sg^2$ in terms of the operations of $\real_{ \bot,\log_2,f}$ is not possible, i.e., a negative answer to the above Problem~\ref{P1}.


\subsection{Kullback-Leibler divergence and Jensen-Shannon divergence}

There are some further information measures worth noting.  Using these primitives, \textit{Kullback-Leibler divergence} (KL divergence for short) can be defined as follows (where an empty sum produces $0$):
$$D_{\mathsf{KL}}(P||Q) =  \sum_{x\in \!V}^{}P(x)\leftmult  \log_2 \dfrac{P(x)}{Q(x)}  $$
or as an expected value:
$$D_{\mathsf{KL}}(P||Q) =  \overline{E}^V_{P}\left(  \log_2 \dfrac{P(x)}{Q(x)}  \right)$$

With  $M_{P,Q}$ as an abbreviation for $ \frac{1}{2}\cdot (P+Q)$,  the \textit{Jensen-Shannon divergence} of $P$ and $Q$ is given by its usual definition.  
 $$D_{\mathsf{JS}}(P||Q) = D_{\mathsf{KL/SO}}(P||M_{P,Q}) + D_{\mathsf{KL/SO}}(Q||M_{P,Q}) $$
Indeed, if $M_{P,Q} (s)= 0$ then $P (s) = Q(s) = 0$ so that 
 $$P(s)\leftmult  \log_2 \dfrac{P(s)}{M_{P,Q}(s)} = 0\leftmult  \log_2 \dfrac{0}{0}= 0\leftmult  \log_2 \bot = 
0\leftmult  \bot = 0$$ and similarly $Q(s)\leftmult  \log_2 \dfrac{Q(s)}{M_{P,Q}(s)}=0$ whence a summand equal to $\bot$ will not arise in the expansion of this particular inslance of Kullback-Leibler divergence.


\section{The algebraic property of flattening}\label{flattening}

Extending fields to common meadows, and then common meadows with further operations opens up many questions about their general algebraic properties. The case of common meadows has been analysed using axiomatisations but here we are in the business of adding several more operators, including familiar ($\log_2, \sg$) and novel operators ($\leftmult$). Also, here we are focussed on algebras of real numbers (not fields in general).

As a tool for the specification of entropy and cross-entropy this paper uses left sequential multiplication $\leftmult$. Is left sequential multiplication of independent interest? We will provide some further facts about it. We will start by recalling the situation for the common meadow $\real_\bot$. 


\subsection{Flattening and common meadows}
In studying meadows, where division is present, we encounter fractions and, in particular, $\frac{1}{0}$ and so infinitely many partial formulae with unclear semantics. In the common meadow, the semantics is settled by $\bot$.  Although the algebra of common meadows disturbs some of the familiar properties of fields, some very important algebraic properties can be preserved.

\begin{definition}
A {\em fracterm} is a term over a signature $\Sigma_{cm}$ for common meadows whose leading function symbol is division $\_/\_$.   A {\em flat fracterm} is a fracterm with only one occurance of the division operator.
\end{definition}

The following is a general fact about meadows (see~\cite{BergstraP2015LNCS}) applied to our meadow of reals:

\begin{theorem}\label{FF}
(Fracterm flattening ) 
For each term $t$ over the  signature $\Sigma$ of the common meadow $\real_{\bot}$ there exist $p$ and $q$ terms over $\Sigma$, both {\em not} involving $\bot$ or division, such that 
$$\displaystyle  t = \frac{p}{q}$$
in $\real_{\bot}$. Furthermore, the transformation is computable.
\end{theorem}

The proof of this is based on a set $E_{\mathsf{ftc-cm}}$ of axioms for common meadows -- see, e.g., \cite{BergstraT2022TCJ,BergstraT2023,BergstraT2024a}. In~\cite{BergstraT2022SACS} we have investigated conditions under which a version of arithmetic with totalized division allows fracterm flattening, with as a tentative conclusion that working in a common meadow, i.e. with $1/0 = \bot$, provides the only option for achieving fracterm flattening, though with some marginal room for variation.


\subsection{Flattening and conditionals}\label{flattening_conditionals}

We notice that the conditional operator allows for fracterm flattening.

\begin{proposition}
Suppose that we extend the common meadow of reals $\real_\bot$ with the conditional operator $x \lhd y \rhd z$ to form $\real_{\bot, x \lhd y \rhd z}$. The algebra $\real_{\bot, x \lhd y \rhd z}$ has fracterm flattening.
\end{proposition}

\begin{proof}
Proofs of flattening for common meadows are derivations from equational axiomatisations \cite{BergstraT2022TCJ,BergstraT2023,BergstraT2024a}.  The new conditional operator satisfies these identities
$$\frac{x}{x'} \lhd y \rhd z = \frac{x \lhd y \rhd z}{x' \lhd y \rhd 1}$$,
$$x\lhd \frac{y}{y'} \rhd z = \frac{(x \lhd y \rhd z)\cdot y'}{y'}$$,
$$x\lhd y \rhd \frac{z}{z'}  = \frac{x \lhd y \rhd z}{1\lhd y \rhd z'}$$
When added to  the  collection $E_{\mathsf{ftc-cm}}$ of axioms for common meadows they can be used as inductive clauses to deal with the extended set of fracterms.
\end{proof}

\begin{lemma}\label{LogCase}
    The algebra $A= \real_{\bot,\,\log_2}$ satisfies the following equation:
    $$\log_2 \frac{x}{y} = \frac{\log_2 x^2 - \log_2 y^2 }{2 + 0 \cdot \log_2 (x \cdot y)}$$
\end{lemma}
\begin{proof}
    If $x= 0$ or $y=0$, then both sides take the value $\bot$. If both $x>0$ and $y>0$ then $ x \cdot y > 0$ and therefore $0 \cdot \log_2 (x \cdot y)= 0$ so that $$\frac{\log_2 x^2 - \log_2 y^2 }{2 + 0 \cdot \log_2 (x \cdot y)}= 
    \frac{2 \cdot \log_2 x - 2 \cdot \log_2 y }{2 }= \log_2 x - \log_2 y= \log_2 \frac{x}{y}$$
    If $x <0$ and $y < 0$ then the same argument applies because $x\cdot y >0$ so that $0 \cdot \log_2 (x \cdot y) = 0$. Finally, if $x > 0 $ and $y < 0$ or $x < 0$ and $y > 0$ then both sides evaluate to $\bot$.
\end{proof}

It follows that fracterm flattening works in $A= \real_{\bot,x \lhd y \rhd z,\log_2}$ as well. Thus, fracterm flattening works for $\leftmult$ in the presence of the conditional operator, but we find that it fails without. 

\begin{proposition} 
Fracterm flattening fails in $A= \real_{\bot,\,\leftmult}$.
\end{proposition}

\begin{proof}
The proof will be given by contradiction. Assuming fracterm flattening for $A$, we may assume that some terms $p$ and $q$ with variable $x$ at most are such that $\sg^2(x)=x\leftmult \frac{1}{x} = \frac{p(x)}{q(x)}$, with $p$ and $q$ division free. 
 
Now,  the numerator and denominator of the fracterm  $\frac{p(x)}{q(x)}$ are simplified by means of repeated application of the following five rules (modulo commutativity and associativity of addition and ordinary multiplication) until no more applications of these rules are possible: 
\begin{center}
$t + \bot \Rightarrow \bot$, 
 
 $t \cdot  \bot \Rightarrow \bot$, 
 
$ - \bot \Rightarrow \bot$, 
 
 $t \leftmult \bot \Rightarrow t $, 
\end{center} 
\noindent and, under the condition  that there is no occurrence of $\bot$ in $r$: 

$$t \leftmult r \Rightarrow t \cdot r.$$
 
\noindent We notice that each of the rewrites
is sound for all $t$ (an argument which depends on the absence of division which might introduce $\bot$ in the absence of an explicit occurrence of $\bot$). 
Suppose that a result of these rewrites is  $p'$ for $p$ and  and $q'$ for $q$. Notice that $q'$ either equals $\bot$, which contradicts the assumption that 
$\sg^2(x)= \frac{p(x)}{q(x)}$ for say $x = 0$, or otherwise $q'$ contains no occurrence of $\bot$ and no occurrence of $\leftmult$ either. It follows that  for all $x$, 
 including $\bot$, 
 $$x\leftmult \frac{1}{x}  = \frac{p(x)}{q(x)} =\frac{p'(x)}{q'(x)}.$$ 
 Now, for $x \neq \bot$ it must be the case that $q'(x) \neq 0$ and 
 by consequence as a function on reals $\frac{p'(x)}{q'(x)}$ is both total and continuous. 
 The latter is not the case for $\sg^2$ so that a contradiction is obtained which proves the failure of 
 fracterm flattening for $\leftmult$.
 \end{proof} 
 If $\sg^2$ is available in combination with $\leftmult$, then fracterm flattening becomes possible as is shown by these equations, both of which are easily validated using case distinctions on $\sg(x)$, $\sg(x')$, $\sg(y)$ and $\sg(y')$:
 $$\sg^2(\frac{x}{y})= \frac{\sg^2(x)}{\sg^2(y)}$$
  $$\frac{x}{x'} \leftmult y = \frac{x \leftmult y}{x'}$$
 $$x \leftmult \frac{y}{y'} = 
 \frac{x^3 \leftmult (y \cdot (y' + 1 - \sg^2(x)))}{(x \leftmult y')^2 + 1 - \sg^2(x))}$$

For basic information on fracterm flattening in the case of common meadows, we refer to~\cite{BergstraT2022SACS}, where it is shown that fracterm flattening for the signature of common meadows almost implies that $\frac{1}{0}$ takes an absorptive value.


\subsection{Fracterm flattening with sign and logarithm}

The algebra $\real_{\bot,\log_2,\sg^2}$ allows expressions of entropy and cross entropy for all finite sizes. Moreover, its default peripheral value $\bot$ is fully absorptive. We notice that $\real_{\bot,\log_2,\sg^2}$
allows fracterm flattening.

\begin{proposition} 
The algebra $\real_{\bot,\log_2,\sg^2}$ allows fracterm flattening.
\end{proposition}  

\begin{proof} The familiar inductive proof for fracterm flattening in  common meadows (including $\real_{\bot}$) is augmented with two cases: $$\sg^2(\frac{x}{y})= \frac{\sg^2(x)}{\sg^2(y)}$$ and, following Lemma~\ref{LogCase}
$$\log_2 \frac{x}{y} = \frac{\log_2 x^2 - \log_2 y^2 }{2 + 0 \cdot \log_2 (x \cdot y)}$$
\end{proof} 

The minimal subalgebra of $\real_{\bot,\log_2,\sg^2}$  is a data  type of importance for probability theory about which we do not have much information, for instance the following problems are left open:

\begin{problem}
Let $A =\real_{\bot,\log_2,\sg^2}$ and let $B$ be the minimal subalgebra of $A$. Is $B$ a computable algebra? If not, is $B$ semicomputable and, if so, is the equational theory of $B$ decidable?
\end{problem}


\section{Concluding remarks}\label{concluding_remarks}

In mapping the problem, we have encountered a number of algebras that have expanded the common meadow of real numbers $\real_{\bot}$, which is our a basic platform for calculation with real numbers, including:
$$\real_{ \bot}, \ \real_{ \bot,\log_2}, \  \real_{ \bot,\log_2,\leftmult}, \real_{ \bot,\log_2, \sg} \ \real_{\bot,\log_2,\sg^2}.$$ 
We will conclude with some comments on expanding fields with operations and peripherals, and a note on probability theory.


\subsection{Fields, meadows, and expansions}
Most useful functions on the real numbers are approximable by polynomials, thanks to Weierstrass's Theorem,\footnote{A continuous function on a compact interval can be approximated to arbitrary precision by a polynomial.} which makes the field of real numbers a truly fundamental algebraic structure for mathematics and its applications. Recall that the field of real numbers has constants $0$ and $1$, and operations $x+y , -x, x \cdot y$.
and that polynomials are normal forms for terms over these operations. However, an approximation of a function is not an exact and complete characterisation (though the approximation process via power series generate one). Many functions are known not to be definable as polynomials, i.e., by terms over fields, including  elementary functions such as $\frac{x}{y}, \sqrt x, \sin(x), \cos(x),\tan (x), e^x, \log(x), \ldots $ etc.
Yet, such special functions appear in important formulae everywhere in pure and applied mathematics. The point is that the neat, understandable, classic formulae and equations of mathematics and its applications are made of terms and that, in turn, they are made from the operations of algebras of reals.

To allow formulae to be constructed that specify functions exactly, special functions can be added as operators to the field operations and composed to make larger classes of terms. The algebraic and logical methods that underpin computing deal with terms made from these operators whose inter-relationships are  analysed axiomatically. This is the approach of abstract data types \cite{EhrichWL1997}.  

An early algebraic appreciation of classes of elementary functions was made by G H Hardy in \cite{Hardy1905}.  Hardy defines an \textit{elementary function} to be a member of the class of functions that comprises:

(i) rational functions

(ii) algebraical functions, explicit or implicit\footnote{Here  explicit or implicit includes root extraction (explicit) and solutions of polynomial equations (implicit).}

(iii) the exponential function $e^x$

(iv) the logarithmic function $log \ x$

(v) all functions which can be defined by means of any finite combination of the symbols proper to the preceding four classes of functions.

Hardy's monograph addresses the questions:
\smallskip

\textit{If f(x) is an elementary function, how can we determine whether its integral is also an elementary function? If the integral is an elementary function, how can we find it?}
\smallskip

Problems of these kinds for different, simpler classes of functions were studied by Richardson who proved that the integration of functions in his class 
is undecidable \cite{Richardson1968}.  Adding operators to fields certainly complicates the logical theories of the extended structures, as Tarski's Exponential Problem for real fields demonstrates.


\subsection{Peripherals for infinity}\label{peripherals4infinity}

Recall the discussion of the peripheral $\bot$ in Section \ref{default_peripheral_value} and its possible replacement by $\infty$. Of course, it is common to write $\infty$ for an infinite value in some extensions of the rational or real numbers. Unlike, $\bot$, the symbol $\infty$ brings with it some widely accepted intuitions.
Suppose a symbol, say $\star$, is a possible `infinite value' then it must satisfy these rules, at least:

(i) $\star$ is a peripheral number, i.e., not an ordinary number

(ii) for real $a > 0$: $a+ \star = \star$

(iii) for real $a > 0$: $a \cdot \star = \star$.\\

\noindent However, these three properties need to be  accompanied by further algebraic properties, which reveal that the intuitions can be diverse and disparate.  Actually, there is a significant variety of semantic options for infinite values, a variety which merits a detailed survey. In our own research programme we have encountered different infinities as follows:

By $\infty_{\mathsf{trp}}$ we denote the positive infinite value of the transreals of~\cite{AndersonVA2007VG}. 

By $\infty_{\mathsf{whu}}$ we denote the unsigned infinite value of wheels (see~\cite{Carlstroem2004MSCS}). 

By $\infty_{\mathsf{etrp}}$ 
we denote the infinite value of the entropic transreals of~\cite{BergstraT2025Entropy}. 

By ~$\infty_\mathsf{strp}$ we denote the peripheral for positive infinity in the symmetric transrationals of~\cite{BergstraT2022LNCS}. 

The value $\bot$ of common meadows meets the conditions (i), (ii), and (iii) above, and so we can consider $\bot$ also as an unsigned infinite value.  In Section \ref{default_peripheral_value} we called it absorptive infinity. Let us denote it  $\infty_\mathsf{cmu}$. Shortly, we will propose another infinity $\infty_\mathsf{cmp}$ for common meadows that serves as a peripheral for a signed positive infinite value.

The signed and unsigned infinities are distinguished by tagging an unsigned infinity with $\mathsf{u}$ and a (positive) signed infinity with $\mathsf{p}$. 

We notice that the following equations are valid:
\medskip

\noindent Common meadows: $\infty_{\mathsf{cmu}} = \frac{1}{0}$, and shortly we will have $\infty_{\mathsf{cmp}} \neq \frac{1}{0}$

\noindent Transreals \cite{BergstraT2020}: $\infty_{\mathsf{trp}} = \frac{1}{0}$ 

\noindent Wheels \cite{BergstraT2021a,Carlstroem2004MSCS}: $\infty_{\mathsf{whu}} = \frac{1}{0}$

\noindent Entropic transreals \cite{BergstraT2025Entropy}: $\infty_{\mathsf{etrp}} = \frac{1}{0}$

\noindent Symmetric transrationals \cite{BergstraT2022LNCS}: $\infty_{\mathsf{strp}} \neq \frac{1}{0}$.
\medskip

Another way of totalising division that is often used is to choose a number as a default value; for this $1/0 = 0$ and $x/0$, are easy choices. This semantic option we have called \textit{Suppes-Ono division} after some early studies \cite{Suppes1957,Ono1983}. What might this option mean for our information measures? Upon expanding an meadow of reals having Suppes-Ono division with a logarithm function, and making the logarithm function total by adopting $\log_2(x) = 0$ for $x \leq 0$ we find that the algebra works well for defining entropy, but not for defining cross-entropy. Thus, Suppes-Ono arithmetic is not sufficiently well-equipped for the description of the basic information theoretic concepts.


\subsection{Adding signed infinities to a signed common meadow}

The common meadow of reals $\real_\bot$ can be further extended with a peripheral $\infty_\mathsf{cmp}$ for a positive infinite value and with its negative counterpart $-\infty_\mathsf{cmp}$. The resulting algebra of reals we denote $\real_{\bot,\infty_\mathsf{cmp}}$.


With the understanding that addition and multiplication are commutative and associative in $\real_{\bot,\infty_\mathsf{cmp}}$, the various operations are defined via these equations:
\medskip

$\infty_\mathsf{cmp} + \infty_\mathsf{cmp} = \infty_\mathsf{cmp} \cdot \infty_\mathsf{cmp} = \infty_\mathsf{cmp}$

$ \infty_\mathsf{cmp} + (-\infty_\mathsf{cmp}) = \bot$

$ \frac{x}{\infty_\mathsf{cmp}} = \bot$ 

$0 \cdot \infty_\mathsf{cmp} = 0$

$a \in \real \to \infty_\mathsf{cmp} +a = \infty_\mathsf{cmp}$

$a \in \real~ \& ~a >0 \to \infty_\mathsf{cmp} \cdot a = \infty_\mathsf{cmp}$

$a \in \real ~\&~ a <0 \to \infty_\mathsf{cmp} \cdot a = -\infty_\mathsf{cmp}$

$\sg(\infty_\mathsf{cmp}) =1$ 

$\sg(-\infty_\mathsf{cmp}) =-1$.
\medskip

The new algebra is no longer a common meadow:

\begin{lemma}
Distributivity is lost in $\real_{\bot,\infty_\mathsf{cmp}}$.
\end{lemma}
\begin{proof}
To see this, calculate:
$$\infty_\mathsf{cmp} \cdot (2-1) = \infty_\mathsf{cmp} \cdot 1 = \infty_\mathsf{cmp}$$
while 
$$\infty_\mathsf{cmp} \cdot 2  - \infty_\mathsf{cmp} \cdot 1 = \infty_\mathsf{cmp} - \infty_\mathsf{cmp} = \bot.$$
\end{proof}
 
\noindent However, fracterm flattening is preserved in $\real_{\bot,\infty_\mathsf{cmp}}$.

What effect might this construction with $\infty_\mathsf{cmp}$ -- rather than $\infty_\mathsf{cmu}$ aka $\bot$ -- have on the formulae for our information measures?

In the presence of logarithm, it is now plausible to adopt 
$$\log_2 0 = -\infty_\mathsf{cmp}$$
instead of $\log_2 0 =  \bot (= \infty_\mathsf{cmu})$, while also maintaining 
$$\log_2 a = \bot \ \textrm{for} \ a < 0.$$ 
So, we obtain the expansion $\real_{\bot,\infty_\mathsf{cmp},\log_2}$ of $\real_{\bot,\infty_\mathsf{cmp}}$.

\begin{lemma}
In $\real_{\bot,\infty_\mathsf{cmp}}$, as $\log_2 \frac{1}{0} = \log_2 \bot = \bot$,  the equation $\log_2 \frac{1}{x} = - \log_2 x$ fails for $x=0$.
\end{lemma}

Thus, we find that the two defining equations for entropy and cross entropy can now be maintained in their original form \textit{without division} (i.e., and without involving sequential multiplication), but the two other forms of defining equations involving division fail for entropy as well as for cross entropy. 

Returning to Example~\ref{example2} we find the following modification of it:
\begin{example}
{\em Suppose $P$ and $Q$ are probability mass functions with the property that for some $x \in V$, $P(x) = \frac{1}{2}$ and $Q(x)= 0$. Then 
 $$-P(x) \cdot \log_2  Q(x) =- \frac{1}{2} \cdot \log_2 0 = -\frac{1}{2} \cdot (-\infty_{\mathsf{cmp}}) = -\infty_\mathsf{cmp}.$$ }
 \end{example}
Signed infinity $\infty_\mathsf{cmp}$ supports the use of two out of four classic defining expressions for entropy and cross entropy, thereby enlarging the common meadow and preserving flattening. In contrast, working in the entropic transreals of~\cite{BergstraT2025Entropy} supports the use of each of the four defining expressions as surveyed in the introduction.


\subsection{Stating the Bayes-Price theorem in common meadows}

Finally, information theory is based on probability theory, which reminds us that elementary probability theory is itself an interesting area for applying algebras with total operators and peripherals. We briefly consider conditional probability and the Bayes-Price Theorem in $\real_\bot$.\footnote{Meadows with Suppes-Ono division allows a plausible development of elementary probability theory (see e.g.,~\cite{Bergstra2019SACS}).}
 
A probability mass function $P$ on a finite event space $S$, which is structured as a Boolean algebra \cite{Padmanabhan1983}, may only have non-negative values different from $\bot$. Now, conditional probability is given by $P(A|B) = \frac{P(A \wedge B)}{P(B)}$. Thus, a conditional probability may take value $\bot$. The Bayes-Price Theorem, commonly and simply stated as $\mathsf{PBT}$:
$$P(A|B) = \frac{P(B|A) \cdot P(A)}{P(B)}$$ 
while paying little attention to additional conditions required for division \cite{Bayes-Price1763}.  

\begin{proposition}
In a common meadow of reals $\real_\bot$, 
$$P(A) \neq 0 \vee P(B) = 0 \to P(A|B) = \frac{P(B|A) \cdot P(A)}{P(B)}.$$
\end{proposition}
\begin{proof}
To see that the formula needs a disjunctive condition
we argue as follows. Note first that if $P(A) \neq 0$ and $P(B) \neq 0$ then $\mathsf{PBT}$ is obvious so that 
the only cases requiring attention are with $P(A) = 0$ and/or $P(B)=0$. 

If $P(A) = 0$ and $P(B) = p >0$ then 
$P(A|B)= \frac{0}{p} = 0$ while 
$$\frac{P(B|A) \cdot P(A)}{P(B)} = \frac{(\frac{0}{0})\cdot 0}{p} = \bot$$
so that the need of making use of a condition at least as strong as $P(A) \neq 0 \vee P(B) = 0$ is manifest. 

Sufficiency of the condition is as follows: (i) if $P(B)= 0$ then for all values of $P(A)$ the required conclusion (i.e., $\mathsf{PBT}$) follows, and (ii) if $P(A) \neq 0$ then the only case that needs to be checked is $P(B) = 0$ which we know already to suffice for the intended conclusion.

\end{proof}


\end{document}